\documentclass[sn-mathphys]{sn-jnl}

\usepackage{lipsum}

\jyear{2022}%

\theoremstyle{thmstyleone}%
\newtheorem{theorem}{Theorem}
\theoremstyle{thmstyletwo}%
\newtheorem{example}{Example}%
\newtheorem{remark}{Remark}%
\newtheorem{corollary}{Corollary}

\theoremstyle{thmstylethree}%
\newtheorem{definition}{Definition}%
\newtheorem{observation}{Observation}
\newtheorem{lemma}{Lemma}

\begin{document}

\title[Geometry of fitness landscapes]{Geometry of fitness landscapes: Peaks, shapes and universal positive epistasis}


\author*[1]{\fnm{Kristina} \sur{Crona}}\email{kcrona@american.edu}

\author[2]{\fnm{Joachim} \sur{Krug}}\email{jkrug@uni-koeln.de}

\author[3]{\fnm{Malvika} \sur{Srivastava}}\email{malvika.srivastava@env.ethz.ch}

\affil*[1]{\orgdiv{Department of Mathematics and Statistics}, \orgname{American University}, 
 \city{Washington DC}, 
 \country{United States}}

\affil[2]{\orgdiv{Institute for Biological Physics}, 
\orgname{University of Cologne}, 
\city{K\"oln},
\country{Germany}}

\affil[3]{\orgdiv{Department of Environmental Systems Science}, \orgname{ETH Z\"urich}, 
\city{Z\"urich}, 
\country{Switzerland}}


\abstract{Darwinian evolution is driven by random mutations, genetic recombination (gene shuffling)
and selection that favors genotypes with high fitness. 
For systems where each  genotype
can be represented as a bitstring of length $L$,
an overview of possible evolutionary trajectories is provided
by  the $L$-cube graph
with nodes labeled by genotypes and edges directed toward the genotype with higher fitness.
Peaks (sinks in the graphs)  are important since a population  can get stranded at a suboptimal peak.
The fitness landscape is defined by the fitness values of all genotypes
in the system. Some notion of curvature is
necessary for a more complete analysis of the landscapes, 
including the effect of recombination. The shape approach
uses triangulations (shapes) induced by fitness landscapes.
The main topic for this work is the interplay between peak patterns and shapes.
Because of constraints on the shapes for $L=3$ imposed by peaks,
 there are in total 25 possible combinations of peak
 patterns and shapes. Similar constraints exist for higher $L$.
 Specifically, we show that the constraints induced by the staircase triangulation can be formulated as a condition of 
 {\emph{universal positive epistasis}}, an order relation on the fitness effects of arbitrary sets of mutations that respects
 the inclusion relation between the corresponding genetic backgrounds.
 We apply the concept to a large protein fitness landscape
 for an immunoglobulin-binding protein expressed in Streptococcal bacteria.}





\maketitle

\section{Introduction}\label{sec1}

This work on acyclic cube graphs and triangulations of cubes is motivated by applications to evolutionary biology. Darwinian evolution can many
times be analyzed by considering  biallelic systems. For a biallelic  $L$-locus system, a genotype $g$ can be represented as a bit string of length $L$. The evolutionary potential for a genotype $g$ is measured by its fitness. 
The fitness landscape $w: \{0,1 \}^L \mapsto  \mathbb R_{\geq 0}$ is determined by the fitness values for all $2^L$ genotypes
\cite{dvk14}. Recent approaches to fitness landscapes rely on analyzing $L$-cube graphs and triangulations that are induced by the
landscapes \cite{Crona2013b}.


One can give a complete (informal) description for $L=2$. The four genotypes are denoted $00, 10, 01, 11$.
The induced graph on the square is determined by the condition that each arrow points toward the genotype of higher fitness
(Figure \ref{Fig:two-locus}). A Darwinian process starting from the genotype $00$ corresponds to a path that respects the arrows.
The graph \ref{Fig:two-locus}C has two peaks (or sinks), which means that an evolving population can get stranded at a suboptimal peak.

The quantity $\epsilon = w_{11}+w_{00}-w_{10}-w_{01}$ measures the deviation of the fitness landscape from additivity
known as epistasis \cite{Domingo2019,Krug2021,Poelwijk2016}. If $\epsilon > 0$ the triangulation
induced by the fitness landscape divides
 the square into two triangles $\{00, 10, 11 \}$ and  $\{00, 01, 11 \}$.
If $\epsilon <0$,  the triangles are instead
$\{00, 10, 01 \}$ and  $\{10, 01, 11 \}$.
For an intuitive understanding, if two genotypes do not belong to the same triangle, 
they could increase their average fitness by swapping positions
(informally, 
$
11+00 \mapsto 10 + 01 
$
increases fitness if  $\epsilon<0$).

In general, we assume that the fitness landscape
is generic in the sense of \cite{BPS:2007}, in particular
that no two genotypes have equal fitness.
The graph determined by the fitness landscape, or the fitness graph \cite{Crona2013a,dvk09},
is the acyclic $L$-cube graph defined by the condition that each edge is
directed toward the genotype of higher fitness. 

Following \cite{BPS:2007}, let $\Delta$ be the simplex  
$
\{ p \in [0,1]^{2^L} : \sum p_g =1 \},
$
where  $p$ can be interpreted as the frequencies of the genotypes in a population,
and $p \cdot w$ measures the average fitness of the  population.
 Let $\rho: \Delta \mapsto [0,1]^L$ be the map  defined as
 $ \left( \rho (p) \right)_i=    \sum _{g_i=1} p_g  $.
 Note that $\rho $ maps gene frequencies to allele frequencies, i.e, 
to  the frequencies of 1's at a particular locus (or string position).
 Define
 \[
 \tilde w( v )=\max_p \{ p \cdot w: \rho (p)= v \} \text{ for all }   v  \in [0,1]^L .
 \]
Then $ \tilde w$ is a piecewise linear function.
The domains of linearity of $ \tilde w$
define a regular triangulation of  $[0,1]^L $ (because $w$ is
generic \cite{triangulations}). The triangulation is called the shape of the fitness landscape.
Recent work on triangulations and fitness graphs
includes \cite{Crona2017,Crona2020a,Crona2020b,Das2020,Eble2019,Eble2020,GouldE11951,Kaznatcheev2019, Lienkaemper2018}.

Fitness graphs and shapes encode information of very different nature,
and can be considered complementary \cite{Crona2013b}.  However, there is also
some overlap in the information. 
The peak pattern for a fitness landscape refers to the number 
of peaks and how they are positioned in the graph, up to cube symmetry.
Some peak patterns impose constraints on the triangulations for $L=2$ (Figure \ref{Fig:two-locus}C)
and $L=3$ \cite{Srivastava}. The main topic for this work is the interplay 
between peak patterns and triangulations. 
For the reader's convenience, a dictionary between terms
used in biology and mathematics has been provided, see Table \ref{notation}.

\begin{figure}
\begin{center}
\includegraphics[width=12cm]{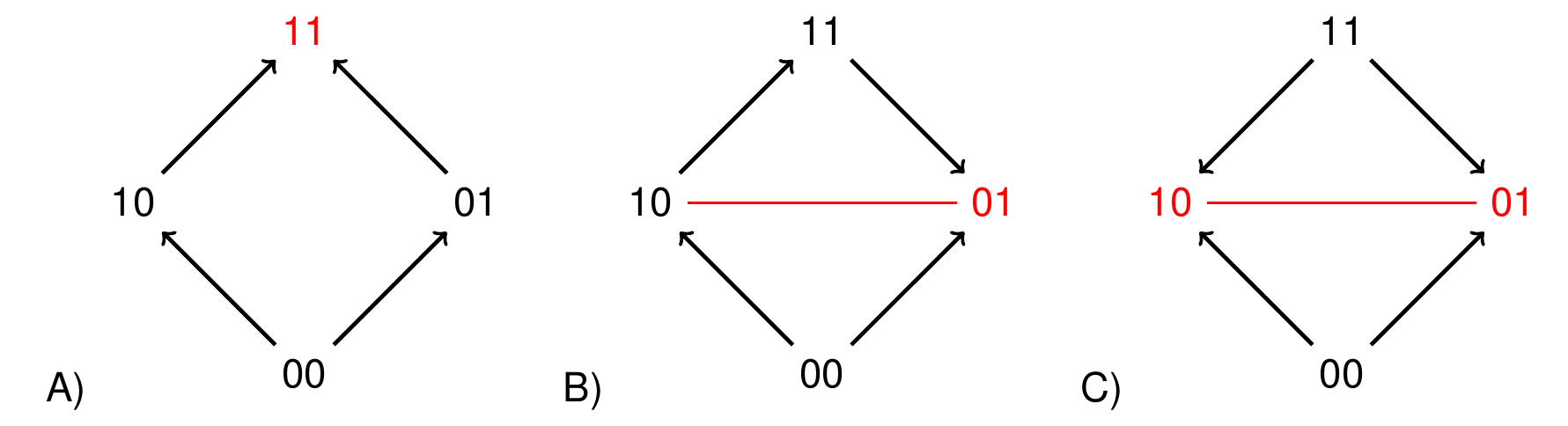}
\end{center}
\caption{Fitness graphs for two loci under the assumption that
$00$ has minimal fitness. If the genotypes are positioned
as in the figure, the three types can be characterized as graphs with all arrows up,  one arrow down, and 
two arrows down. Peak genotypes are marked in red. Graphs B) and C) display sign epistasis \cite{Weinreich2005}, which
means that at least one pair of arrows on parallel edges point in opposite directions. In graph C) sign epistasis
is reciprocal \cite{Poelwijk2007}.
By Observation \ref{observation:diagonals},
graphs B) and C) are compatible only with the triangulation indicated
by the red lines. Note that for general fitness graphs any genotype can have minimal fitness.}
\label{Fig:two-locus}
\end{figure}


\vspace*{0.5cm}

\noindent
{\bf{Main Results.}}
For  generic fitness landscapes $w: \{0,1 \}^L \mapsto  \mathbb R_{\geq 0}$
we compare the induced peak patterns and triangulations.
For $L=3$ we show that there are exactly 25 possible 
combinations of peak patterns and triangulations (\textit{Theorem \ref{theoremfirst}}).
The statistical distribution of peak patterns conditioned on the triangulation is obtained from simulations of random
fitness landscapes. 
For higher  $L$ we show that some peak patterns
are compatible with all triangulations,
whereas other peak patterns are incompatible
with almost all triangulations. A peak pattern and a triangulation are referred to as
compatible if there exists a fitness landscape that induces both of them.

Additional results for general $L$ can be obtained for fitness landscapes that induce staircase 
triangulations (Definition \ref{staircase}).
We show that these landscapes display a property that we call \textit{universal positive epistasis}.
The property holds if 
\begin{equation}
  \label{Positive_epistasis}
w_{g \cup g'} + w_{g \cap g'} \geq  w_{g} + w_{g'}
\end{equation}
for  all  pairs of genotypes $g$ and $g'$ interpreted as sets of 1-alleles. 
It implies that the fitness effect of a given set of mutations on different genetic backgrounds
inherits the partial order induced by the subset-superset relation between background genotypes.
The conditions \eqref{Positive_epistasis} characterize the (standard) staircase triangulation.
Universal positive epistasis limits the maximal
number of peaks in the fitness landscape, and explicit upper bounds are derived for $L=4$ and $L=5$
(\textit{Theorem \ref{theoremsecond}}).
Gr\"obner bases for staircase triangulations (Section \ref{circuits})
were used for establishing that the bound for $L=4$ is sharp.
The analysis of a large empirical data set of fitness interactions between protein substitutions selected for their
positive epistatic effects \cite{Wu2016} shows a significant overrepresentation of the staircase triangulation. 
  
\section{Peak patterns and triangulations for $L=3$}
With notation as in the introduction, we consider
fitness graphs and triangulations induced by fitness landscapes. 
The peak set of a fitness graph determines
its {\emph{peak pattern}}.
Two fitness graphs have the
same peak pattern if for instance
the set of peaks differ by a cube rotation 
only (see below for a more formal discussion).
In particular, all graphs with a single peak have the same
peak pattern.

\begin{figure}[t]
\begin{center}
\includegraphics[width=0.6\textwidth]{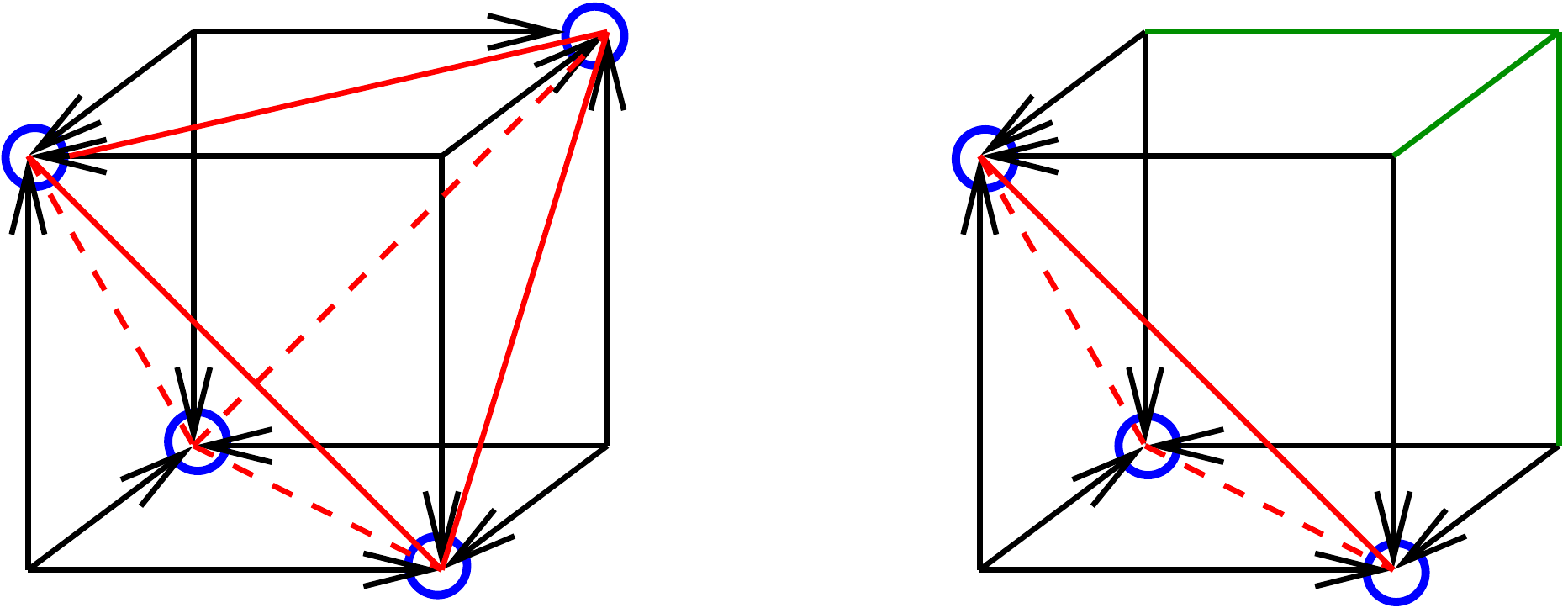}
\end{center}
\caption{Three locus fitness graphs with 4 and 3 peaks.
Face triangulations are indicated by full red lines for the exposed faces and dashed red lines for the hidden
faces of the cube. Left panel: For the graph with 4 peaks, the peak pattern fully specifies
the triangulation of the faces. The triangulation has to be of type 1 or 2 (these types 
have identical face triangulations, see Figure 4). 
Right panel: 
The presence of three peaks implies that a corner is isolated,
which is inconsistent with type 6 triangulations.}\label{Fig:34peaks}
\end{figure}

\begin{figure}[t]
\begin{center}
\includegraphics[scale=0.4]{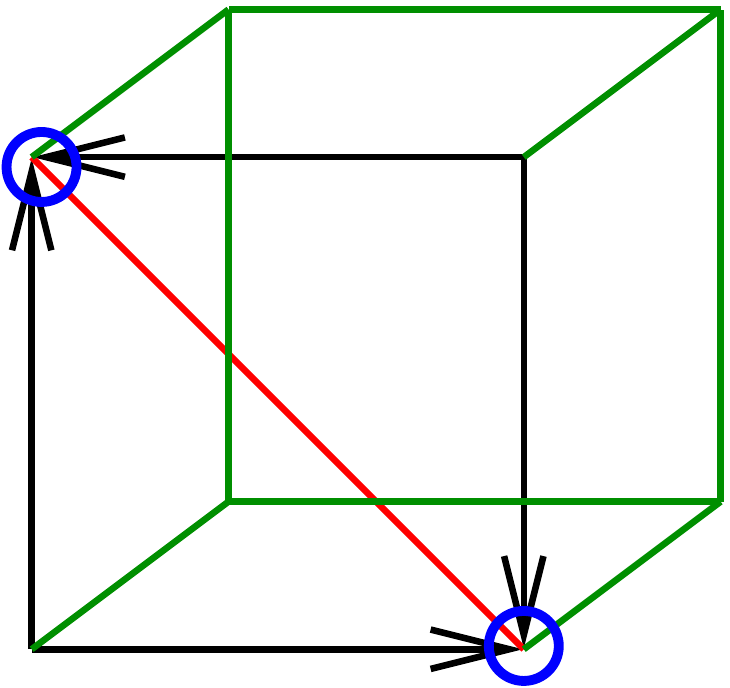}
\qquad
\includegraphics[scale=0.4]{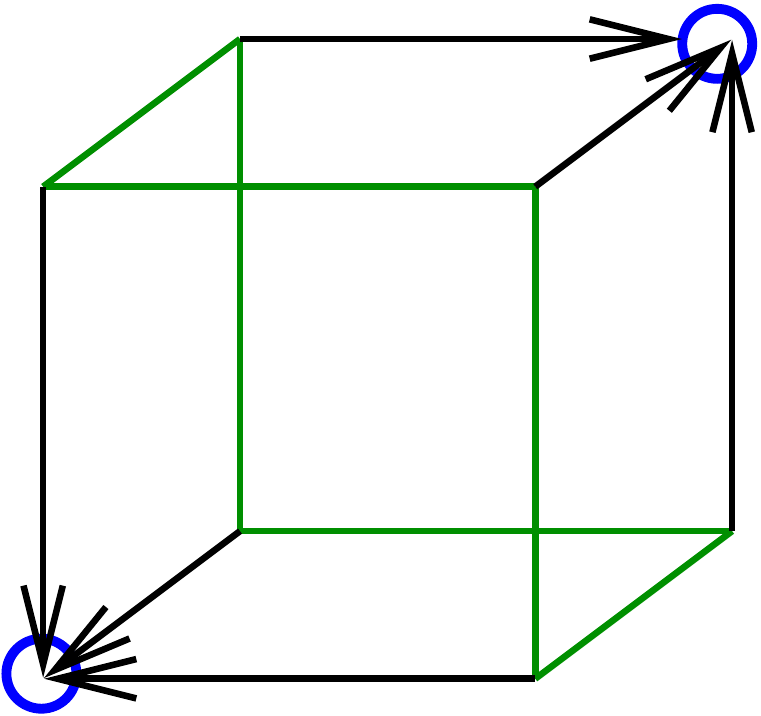}
\end{center}
\caption{For graphs with two peaks, there are two possible peak patterns, depending on if the distance between the peaks is two or three. 
If the distance is two, the face triangulation is as indicated by the red line. 
Both cases are compatible with all six triangulations.}
\label{Fig:2peaks}
\end{figure}

\subsection{Classification results for $L=3$}
\begin{observation}
For $L=3$ there are 5 peak patterns, i.e., four patterns
(Figures \ref{Fig:34peaks} and \ref{Fig:2peaks}) in addition
to the case with a single peak.
\begin{itemize}
\item Fitness graphs have 1-4 peaks.
\item The peak pattern is unique if the number of peaks is 1, 3 or 4.
\item There are two peak patterns  if the number of peaks is 2.
The distance between the peaks is either 2 or 3.
\end{itemize}
\end{observation}

A triangulation of the cube can be described as a subdivision
of the cube into tetrahedra. 
It is well known that in the case of the 3-cube only 6 triangulations can arise, 
up to symmetries, see \cite[Thm.~6.3.10]{triangulations}. The possible shapes
are illustrated in Figure 4.
The graphs below the cubes aid visual interpretations of the triangulations.
The nodes of the graphs represent tetrahedra.
Two vertices are joined by an edge
if the corresponding tetrahedra have a triangle in common,
and a shaded region indicates that the 
 tetrahedra share a line segment.
The graphs are defined as tight spans dual to the triangulations \cite{Herrmann2014}.
The corner tetrahedra (see the top row in Figure 4)
appear as nodes of degree 1 in the tight spans (i.e., nodes with only one outgoing edge).

\begin{figure}
\begin{center}
\includegraphics[width=\textwidth]{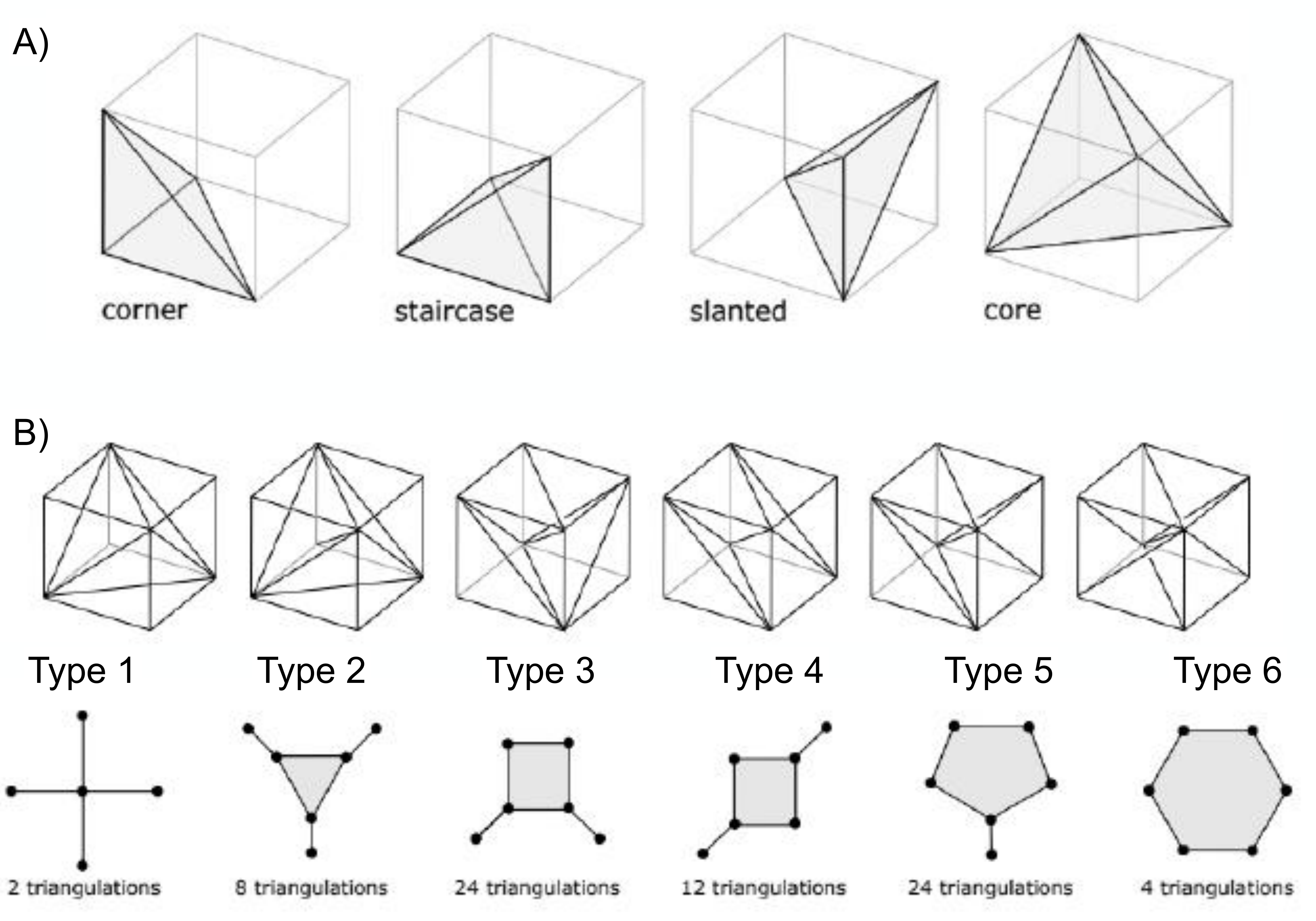}
\caption{The six types of triangulations of the 3-cube, together with the tight spans dual to the triangulations (modified from \cite{Pellerin2018} with
permission).  Row A) shows the types of simplices that occur in the triangulations.
}
\end{center}
\label{Fig:triangulations}
\end{figure}

\begin{theorem}\label{theoremfirst}
The five peak patterns described in the previous observation, and the six triangulation types, as  described in
Figure 4, 
can be combined in 25 ways. Specifically
\begin{itemize}
\item Fitness graphs with 4 peaks are compatible with triangulations of type 1 and 2, but not with any other triangulations.
\item Fitness graphs with 3 peaks are compatible with triangulations of type 1-5, but not with type 6.
\item The remaining three peak patterns are compatible with all six triangulation types.
\end{itemize}
\end{theorem}

Before giving a proof we introduce some formal notation for general $L$.
The peaks of a fitness graph define a peak pattern for any $L$.
Formally, a peak pattern is an equivalence class of
peak sets under the action of the hyperoctahedral group 
of cube symmetries (the group has order $2^L L!$). Peak patterns have been enumerated for $L$-cubes up to $L=6$ \cite{Oros2022}
(see Table \ref{table:peakpatterns}). 

\begin{table}\centerline{
  \begin{tabular}{c|cccccc}
    $L$ & 2 & 3 & 4 & 5 & 6 \\
    \hline
      regular triangulations & 2 &  74 &  87,959,448 & -  & -  \\
      \hline
      symmetry classes & 1 & 6 & 235,277 & - & -  \\
       \hline
      peak patterns & 2 & 5 & 20 & 287 & 519,194       
  \end{tabular}}
  \caption{The table summarizes known results for the number of regular triangulations and their symmetry classes
    \cite{triangulations,MR2398786}, and the number of peak patterns \cite{Oros2022} for $L$-cubes.}\label{table:peakpatterns}
\end{table}

A triangulation of the $L$-cube is a subdivision
of the cube into $L$-simplices (i.e., simplices that can be described as
the convex hull of $L+1$ affinely independent points),
such that any non-empty intersection between 
simplices are faces of both of them.
The triangulation induced by the fitness landscape $w$ was defined in the introduction.
In the terminology of \cite{triangulations}, 
the triangulation is the projection of 
the upper envelope of the convex hull of the genotypes lifted by $w$.
(For clarity, each genotype $g$ determines a point in $\mathbb R^{L+1}$ by 
regarding its fitness $w_g$ as a height coordinate. The triangulation
is the projection of the upper faces of the polytope 
obtained as the convex hull of such points).
Triangulations induced by fitness landscapes are defined 
as regular triangulations. There exist non-regular triangulations 
for $L \geq 4$ \cite{triangulations,pournin2013}, but such triangulations are not considered in this
paper.

For $L=3$, a triangulation divides each side of the cube along a 2-face diagonal,
i.e., the diagonal determined by vertices $g, g'$ such that
$ \{g, g'\}$ belongs to a tetrahedron in the triangulation. Such diagonals
are  referred to as induced diagonals.
Similarly, for general $L$ a triangulation induces a 2-face 
diagonal of each 2-face.
Specifically, consider a 2-face where the two pairs
 of vertices with distance two represent  genotypes $\{g, g'\}$
 and $\{h, h'\}$. The $g,g'$-diagonal is
 induced by the triangulation  if
 $w_{g}+w_{g'}> w_h+w_{h'}$, which implies that $h$ and $h'$ 
 do not belong to the same simplex in the triangulation. The following
 observation specifies how induced diagonals are constrained by the fitness
 graph (see Fig.~\ref{Fig:two-locus} for illustration).
\begin{observation}\label{observation:diagonals}
Consider a 2-face composed of pairs of vertices  $\{g, g'\}$ and $\{h,
h'\}$ with distance two. Suppose the fitness graph displays sign
epistasis on the 2-face, i.e. at least one of the two pairs of
arrows on parallel edges point in opposite directions. Then the
induced diagonal connects the arrow heads. 
\end{observation}
\begin{proof} Assume w.l.o.g. that the antiparallel arrows point from
  $h'$ to $g$ and from $h$ to $g'$. This implies that $w_g > w_{h'}$
  and $w_{g'} > w_h$. Hence $w_{g}+w_{g'}> w_h+w_{h'}$ and the induced
  diagonal connects $g$ and $g'$.
  \end{proof}

Following \cite{triangulations} we say that the vertex with the central position
in a  corner simplex is sliced off by the triangulation
(see the bottom left vertex of the corner simplex in Figure 4). 
  For $L=3$ a sliced off vertex belongs to exactly one
tetrahedron. In general, a sliced
off vertex of the $L$-cube belongs to exactly one simplex of the triangulation (the
simplex consists of the vertex itself and its $L$ neighbors).


The proof of  Theorem \ref{theoremfirst}  includes 
an analysis of 2-face diagonals 
induced by the triangulation.
The diagonals marked red in Figures \ref{Fig:34peaks} and \ref{Fig:2peaks}
are induced by the triangulation, because they connect peaks.
For the proof it is  helpful to introduce a new concept.
For arbitrary $L$, we call a genotype $g$ isolated if there is no genotype 
$g'$ such that $\{g, g'\}$ belongs to the same simplex and $\| g-g' \|=2$,
i.e., no induced 2-face diagonal connects $g$ to another genotype.

\begin{remark}
  \label{remark:slicedoff}
If the triangulation slices off a genotype, 
then the genotype is isolated. However,
the converse is not true (see e.g. the Type 2 triangulation in Figure 4).
\end{remark}

\begin{observation} 
For the six triangulation types described in Figure 4,
\begin{itemize}
\item Type 1 triangulations have four isolated vertices.
\item Type 2 triangulations have four isolated vertices.
\item  Type 6 triangulations have no isolated vertices.
\item All the remaining triangulation types have one or two isolated vertices.
\end{itemize}
\end{observation} 

The proof of Theorem \ref{theoremfirst}
relies on the previous observation and
simulations summarized in Figure \ref{fig:Malvika1} 
(see Section \ref{Sec:Statistics} for details on the simulations).

{\emph{Proof of Theorem  \ref{theoremfirst}.}
If the fitness graph has 4 peaks, then the remaining four
genotypes are isolated, which excludes the triangulations of types 3-6 by the previous observation.
If the fitness graph has three peaks, then the genotype adjacent to the three peaks is isolated,
which excludes the type 6 triangulation.
Out of the 30 combinations of peak patterns and triangulations, 5 are excluded
as argued. 
The  remaining 25 combinations of peak patterns and triangulations appear with nonzero probability
in statistics for randomly generated fitness landscapes (see Figure \ref{fig:Malvika1}). \hfill \qed

\begin{remark}
There is an overlap between  Theorem \ref{theoremfirst} and
results in  \cite{Seigal2018}, specifically concerning the case with
four peaks (see Corollary 1.4 of  \cite{Seigal2018}).
\end{remark}

\begin{remark}
Theorem \ref{theoremfirst} describes peak patterns
and triangulations that are compatible.
A related question concerns if it is possible  to
combine a particular peak pattern and
triangulation in more than one way.
Classification results on such combinations would
be of interest, but the topic is beyond the scope of this paper.
\end{remark}

\smallskip
\noindent
    {\bf{Open question:}} For the $L=4$-cube there are 20 peak patterns and 235,277 symmetry classes of regular triangulations
    (Table \ref{table:peakpatterns}). Which combinations of peak patterns and triangulations are compatible?

\bigskip

 \subsection{Peaks and shapes of random fitness landscapes} \label{Sec:Statistics}
 Generic random fitness landscapes are obtained by assigning independent, identically and continuously distributed random
 numbers to the genotypes \cite{Kauffman1987}. Because peak patterns and fitness graphs are fully determined
 by the rank order of fitness values, the resulting statistics are independent of the underlying probability distribution.
 A simple rank order argument shows that the expected number of peaks of a random landscape over $L$ loci is $\frac{2^L}{L+1}$,
 and the variance of the number of peaks is also known \cite{Macken1991}. The full distribution of peak patterns for $L=3$
 was obtained in \cite{SK2014} and is reported in Table \ref{tab:summary_peaks}. Additionally the table shows the distribution
 of peak patterns over the 54 isomorphism classes of fitness graphs presented in \cite{Crona2017}. 

 \begin{center}
\begin{table}
\renewcommand{\arraystretch}{1.6}
  \begin{tabular}{l l l}
    \hline
	 Peak pattern  &  Random landscapes & Isomorphism classes  \\ \hline
	        1 peak  &   $\frac{3}{14} \approx 0.214$ & $\frac{1}{2} = 0.5$  \\ 
	        2 peaks at distance 2   &  $\frac{5}{14} \approx 0.357$ & $\frac{5}{18} \approx 0.278$    \\ 
	        2 peaks at distance 3   &   $\frac{1}{4} = 0.25$ & $\frac{4}{27} \approx 0.148$    \\ 
	        3 peaks    &  $\frac{1}{7} \approx 0.143$ & $ \frac{1}{18} \approx 0.055$    \\ 
	        4 peaks (Haldane graph)  &   $\frac{1}{28} \approx 0.036$ & $\frac{1}{54} \approx 0.0185$      \\ \hline
 \end{tabular}
  \caption{Distribution of peak patterns for random three-locus landscapes based on \cite{SK2014} and \cite{Crona2017}. Decimal
    approximations of the exact rational probabilities are provided for convenience. }\label{tab:summary_peaks}
 \end{table}
\end{center}

Figure \ref{fig:Malvika1} shows the distribution of peak patterns conditioned on the triangulation type. 
The combinations of peak patterns and triangulation types excluded by Theorem \ref{theoremfirst} are 
absent in the figure, whereas all other combinations occur with positive probability.
In contrast to the statistics of fitness graphs presented in Table \ref{tab:summary_peaks},
these statistics depend on the probability distribution of
fitness values \cite{Srivastava}. 
Observations similar to  the summary in Figure \ref{fig:Malvika1}, but focused on the number of peaks
rather than peak patterns, appeared first in  \cite{Srivastava}.
The results in Figure  \ref{fig:Malvika1} were obtained using the uniform distribution on $[0,1]$ (see
Figure~\ref{Fig:sublandscapes} for the distribution of triangulation types in this case).
Furthermore, it can be seen that the overall ruggedness of the landscapes (as measured, e.g.,
by the mean number of peaks) decreases systematically going from triangulation type 1 to type 6. A pronounced
trend is observed for the patterns with two peaks: Whereas for triangulation type 1 the two peaks are almost exclusively at
distance 2, for triangulation type 6 the pattern with two peaks at distance 3 dominates. 

\begin{figure}[ht]
\begin{center}
\includegraphics[width=0.8\textwidth]{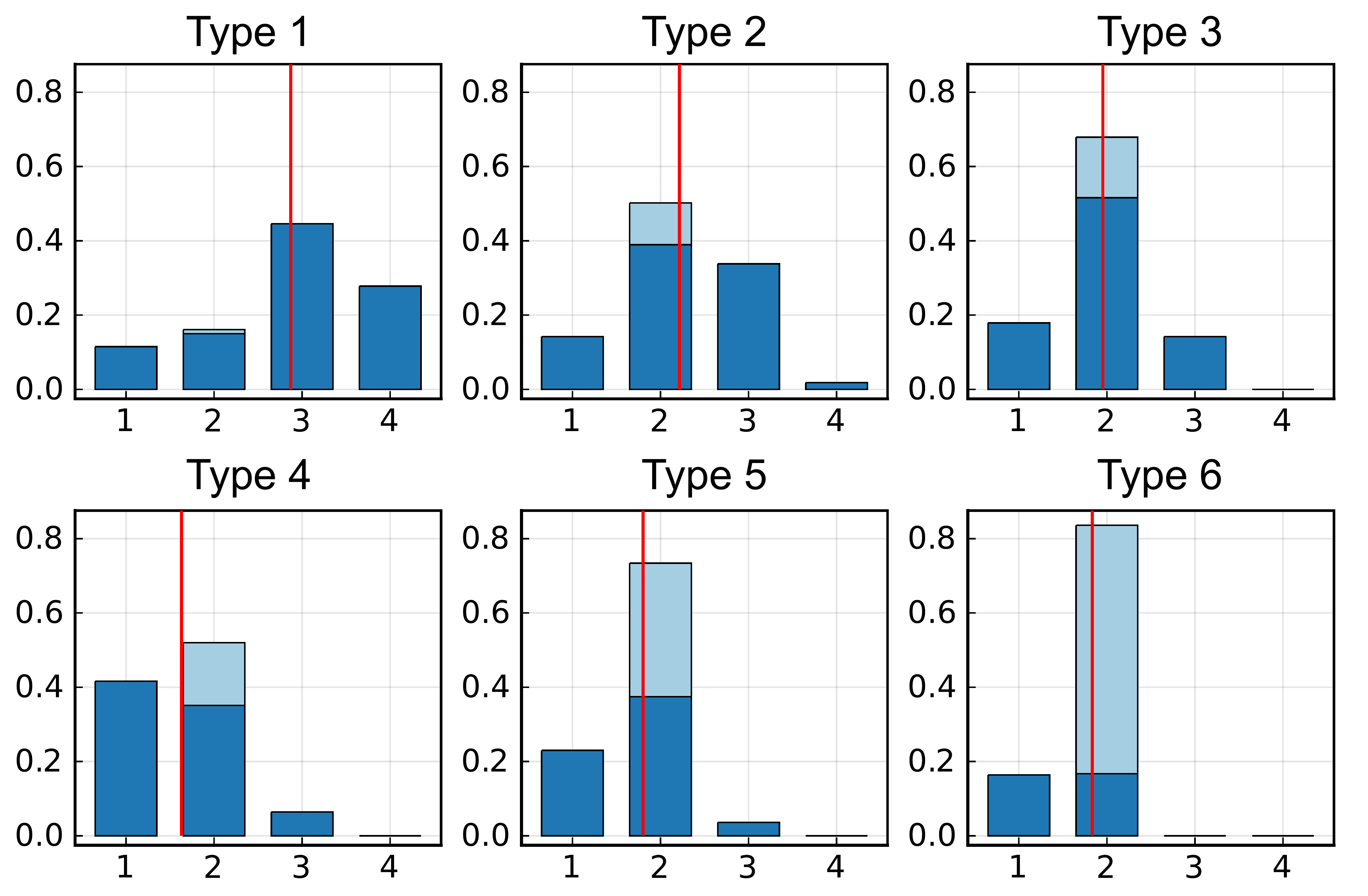}
\end{center}
\caption{Distribution of the number of peaks for different triangulation types.
 The fitness values were taken to be uniformly distributed on $[0,1]$. The vertical red lines indicate the mean of the
distribution. For unconstrained random three-locus landscapes the mean number of peaks is 2. The light-blue portion of the bars
at peak number 2 shows the fraction of landscapes where the two peaks are at distance 3.}
\label{fig:Malvika1}
\end{figure}


\section{Higher dimensional fitness landscapes}
Throughout the following sections, all triangulations are assumed to be regular.
The regularity assumption will not be stated explicitly.

\subsection{Extreme peak patterns and triangulations}
A counting argument shows that the maximal number of peaks for a fitness graph
is $2^{L-1}$. This fact was first observed in \cite{haldane_1931}, and we refer to these graphs as Haldane graphs.
One can verify that each node in a Haldane graph is either
a sink or a source.

\begin{remark}
A  fitness graph has at most  $2^{L-1}$ peaks.
For graphs with the maximal number of peaks, each node
is either a sink (peak) or a source. All neighbors of sinks are sources,
and vice versa.
\end{remark}

Among graphs with only one peak, the all arrows up graph
is defined as the graph where fitness decreases by distance
from the node with maximal fitness (if one draws the graphs 
as in Figure \ref{Fig:two-locus}, all arrows point upward).

\begin{observation}
The all arrows up graph is compatible with all triangulations.
\end{observation}

\begin{proof}
The argument relies on the following construction:
For a fitness landscape $w$ and a constant $c\in\mathbb N$ let
\[
w_{g}^{+c}=w_g +c \sum_{i=1}^L g_i .
\]
For instance,
$w_{{00}}^{+2} =w_{{00}}$,  $w_{{10}}^{+2} = w_{{10}}+ 2$, $w_{{01}}^{+2}=w_{{01}}+2$, $w_{{11}}^{+2}=w_{{11}}+4$.
Consider an arbitrary triangulation of the $L$-cube induced by $w$.
One verifies that the fitness landscapes $w^{+c}$ and $w$ induce the same triangulation.
For sufficiently large $c$, the fitness graph  corresponding to $w^{+c}$ is an all arrows up graph,
which completes the proof.
\end{proof}

\begin{corollary}
All triangulations are compatible with some single peaked fitness landscapes.
\end{corollary}

Type 1 and type 6 triangulations can be considered opposite extremes for $L=3$.
The first triangulation slices off four genotypes, i.e., it has four corner
simplices, and the latter does not slice off any genotypes. 
For general $L$, a corner simplex consists of a vertex $g$ (the sliced
off vertex) and its $L$ neighbors \cite{triangulations}.
Analogous to the $L=3$ case, 
 corner-cut triangulations are defined as 
 triangulations that have   $2^{L-1}$  corner simplices, i.e., the
 highest number of corner simplices for a given $L$  \cite{SCOTTMARA1976170}.
Corner-cut triangulations are unique  for $L=3$ and $L=4$
but not for general $L$ \cite{triangulations}.

The  staircase triangulation for general $L$ is a direct generalization of 
the type 6 triangulations. For notational convenience
we define the standard staircase triangulation so that
each simplex contains the genotypes from a walk of length $L$ from 
 $00 \dots 0$ to $11\dots1$.
 All $L!$ walks are represented.
For instance, for $L=3$ the triangulation is 
\begin{align*}
&\{000, 100, 110, 111\}, \{000, 100, 101, 111\}, \{000, 010, 110, 111\}, \\
& \{000, 010, 011, 111\} , \{000, 001, 101, 111\}, \{000, 001, 011, 111\}.
\end{align*}

 \begin{definition}\label{staircase}
The  standard staircase triangulation  consists of  $L!$ simplices,
such that each simplex contains all genotypes from
a walk of length $L$, starting at  $00 \dots 0$ (the zero-string) 
and ending at the one-string $11\dots 1$, where the number of $1$'s increases in each step. 
Any triangulation that is isomorphic
to the standard staircase triangulation is
referred to as a staircase triangulation.
\end{definition}
In particular, there are four isomorphic staircase triangulations
for $L=3$ since the pair of vertices that belong to all tetrahedra
can be chosen as $\{000, 111\}$,
$\{100,  011\}$, $\{010,  101\}$ or $\{001, 110 \}$.

Also for general $L$, it is useful to consider the induced
triangulations of 2-faces.
Diagonals connect peaks on the 2-faces for
triangulations compatible with Haldane graphs.
It follows that the source genotypes are isolated.
However, they are not sliced off in general (Remark \ref{remark:slicedoff}).

\begin{remark}\label{source}
The source genotypes are isolated for Haldane graphs.
\end{remark}

Consider the standard staircase triangulation for $L\geq 3$. 
For any genotype $g$, let  $g'$ be a genotype
obtained by changing two loci (positions) from $1$ to $0$, 
or two loci from $0$ to $1$. Then  $g$ and $g'$ belong to the same simplex, 
which implies that $g$ is not isolated.

\begin{remark}
  \label{remark:staircase}
If the fitness landscape induces a staircase triangulation, then there are no isolated genotypes.
\end{remark}

\begin{observation}
Haldane graphs impose restrictions on triangulations
which make them incompatible with some triangulation types.
The graphs are compatible with other well studied triangulation types.

\begin{enumerate}
\item[(i)]
Triangulations that are compatible with Haldane graphs have $ 2^{L-1} $ isolated genotypes.
\item[(ii)]
Haldane graphs are incompatible with staircase triangulations for $L \geq 3$.
\item[(iii)]
For any $L$, there exists a corner cut triangulation that is  compatible with the Haldane graph.
\end{enumerate}
\end{observation}

\begin{proof}
(i) follows from Remark \ref{source}.
(ii) follows from (i) and Remark \ref{remark:staircase}.
For (iii), consider a generic fitness  landscape that induces a Haldane graph such that  $w_g \approx 2$ for half of the genotypes 
and  $w_g \approx 1$ for the remaining genotypes. One can verify that the fitness landscape induces a  corner-cut triangulation
provided the approximations are sufficiently close (each source genotype is sliced off).
\end{proof}

We have established that a Haldane graph is compatible with 
some corner-cut triangulation. A stronger claim would be that
any corner-cut triangulation is compatible with the Haldane graph.

\smallskip
\noindent
{\bf{Open question:}} Are all corner-cut triangulations compatible with Haldane graphs?

\subsection{Staircase triangulations and universal positive epistasis}\label{subsect:staircase}
It is sometimes convenient to use the set notation for genotypes \cite{Das2020}. 
If the bitstring representing the genotype has a 1 in position $i$, then the set contains the element
$i \in {\mathcal{L}}$ of the locus set ${\mathcal{L}} = \{1,2,\dots,L\}$.
For $L=2$, the translation is $00=\emptyset, 10=\{1\}, 01=\{2\}, 11=\{1,2\}$.

We consider the standard staircase triangulation as previously defined
(Definition \ref{staircase}).
For a pair of genotypes $g, g'$, 
let
 \[
 X_{g,g'}(i)=\frac{1}{2} (g_i + g'_i).
 \] 
 In the terminology of \cite{BPS:2007},
 $X_{g,g'}$ is the allele frequency 
vector of a population composed
in equal part of $g$ and $g'$, and
such a population has fitness
$
\frac{1}{2} ( w_g + w_{g'} ).
$

\begin{observation}
Assume that $w$ induces the standard staircase triangulation.
Let  $g, g'$ be two genotypes where $g' \subset g$.
Consider all pairs of genotypes $\{h, h'\}$
such that   
\[
X_{h, h'}=X_{g,g'}.
\]
Then
\[
w_g + w_{g'} \geq w_h+w_{h'},
\]
with equality only if the pairs $\{h, h'\}$ and $\{g,g' \}$ are identical.
\end{observation}

\begin{proof}
Assume that  $g'\subset g$.
Then $g$ and $g'$ belong to the
same simplex in the standard staircase
triangulation.
Suppose that $\{h,h'\}$ is a pair of
genotypes such that
$X_{g,g'}=X_{h,h'}$.
It is easy to verify that 
$h' \cup h=g$ and $h' \cap h=g'$.
Unless the pair $\{h', h\}$ is identical to  $\{ g, g' \}$, one can 
exclude that $h' \subset h$ or $h \subset h'$,
which  shows that  $h$ and $h'$ 
do not belong to the same simplex
in the triangulation.
From these observations
and properties of the fittest population
described  in \cite{BPS:2007}
the result is immediate.
\end{proof}

To make contact with the usual definition of epistasis \cite{Domingo2019,Krug2021,Poelwijk2016},
let  
$b' \subset b \subset \mathcal{L} $, and let
$s \subset  \mathcal{L} \setminus b$.
The previous observation
applied to  the genotypes $g=b \cup s$ and $g'=b'$ implies that
\begin{align}\label{eq:posep}
	w_{b \cup s} - w_{b} \geq w_{b' \cup s} - w_{b'}.
\end{align}
In this formulation a set $s$ of loci is added to two different backgrounds, 
$b$ and $b'$, and the fitness effect of this substitution is larger in $b$ than in $b'$ if $b' \subset b$.

\begin{definition}[\emph{Universal positive epistasis.}]\label{universal}
	A fitness landscape displays universal positive epistasis if 
	for any genotypes  $b' \subset b \subset \mathcal{L} $ and
	set $s \subset  \mathcal{L} \setminus b $,
	the inequality \eqref{eq:posep} holds.
\end{definition}

Equivalently, universal positive epistasis holds if the condition (\ref{Positive_epistasis})
is satisfied for any two genotypes $g$ and $g'$
or, if genotypes are represented as strings
\[
w_{g \lor g'} + w_{g \wedge g'} \geq  w_{g} + w_{g'}.
\]
Informally, for a given allele frequency vector, one can construct the pair of genotypes with maximal
fitness sum by distributing the 1's as unevenly as possible. For instance
\[
w_{111}+w_{000}> w_{110} + w_{001},  \, w_{101}+w_{010}, \,  w_{011}+w_{100} .
\]
Note that the allele frequency vector is $[1/2, 1/2, 1/2]$ for the four pairs of genotypes,
and that all of them except $\{111, 000\}$ consists of two genotypes that belong to different
tetrahedra in the triangulation.
Conversely, a fitness landscape with 
universal positive epistasis  induces the standard
staircase triangulation (see Observation \ref{positive}).

\begin{remark}
The universal positive epistasis condition is a characterization of  
 the standard staircase triangulation.
\end{remark}

 \subsection{Peak patterns for the staircase triangulation}
In this section single, double and triple mutants, and similarly, refer to
the cardinality of the sets used in the set notation for genotypes. The zero-string
is represented by the empty set.

 In the case of $L=3$, fitness landscapes compatible with the
 staircase triangulation were found to have few peaks, and here we will show that this observation can be extended
 to general $L$. In particular, local conditions suffice to ensure that the landscape is single peaked.

 \begin{observation} For fitness landscapes that induce the standard staircase triangulation
   the following statements hold: \\
   (i) If the zero-string is a local fitness minimum, it is also the global minimum and the fitness graph has all arrows up. \\
   (ii) If at most one of the $L$ neighbors of the zero-string has lower fitness, the fitness landscape is single-peaked. Moreover,
   the peak is at most at distance 1 from the one-string.
 \end{observation}
 \begin{proof} (i) By assumption $w_i > w_\emptyset$ for all $i \in {\mathcal{L}}$. Since $\emptyset$ is a subset of all genotypes,
 it follows from (\ref{eq:posep}) that $w_{b \cup i} - w_b > 0$ for all background genotypes $b$. Thus all mutations are beneficial
 on all backgrounds and the fitness graph has all arrows up.  \\
 (ii) Suppose $w_j < w_\emptyset$ for some $j \in {\mathcal{L}}$
 and $w_i > w_\emptyset$ for all $i \neq j$. Then $w_{b \cup i} - w_b > 0$ for all $i \neq j$. This implies that a peak genotype
 can have a zero entry only at position $j$, because otherwise the fitness could be increased by changing the entry to 1.
 The possible peak genotypes are the one-string and its neighbor with a zero entry at position $j$. They cannot both
 be peaks. 
 \end{proof}

The following result is useful for counting peaks.

\begin{lemma}\label{prop}
Let $w$ be a fitness landscapes that induces the standard staircase triangulation.
\begin{itemize}
\item[(i)]
If a single mutant $\{ i \}$ is a peak, then all  other peaks contain $\{ i \}$.
In particular, two single mutants cannot both be peaks.
\item[(ii)]
Two double mutants $\{ i,j \}$ and $\{ j, k \}$ that overlap cannot both be peaks.
\item[(iii)]
A triple mutant $\{i,j,k\}$ and a double mutant $\{k,l\}$ that overlap cannot both be peaks.
\end{itemize}
\end{lemma}

\begin{proof}
(i): Suppose that $\{ i \}$ and $g$ are peaks, where $i  \notin g$.
Then
$
w_g>w_{g \cup  \{ i \}}
$
and
$
w_i> w_\emptyset
$.
It follows that
\[
w_{g \cup \{ i \}} + w_\emptyset <
w_i +w_g,
\]
which contradicts \eqref{Positive_epistasis}.

(ii): Two double mutants $\{ i,j \}$ and $\{ j, k \}$ that overlap cannot both be peaks.

Suppose that $w_{ij}$ and $w_{ jk}$  are peaks,
then 
$
w_{ij}> w_j
$
and
$
w_{jk}> w_{i jk}.
$
It follows that
\[
w_{i j k }+w_{ j }  < w_{i  j} + w_{j k}
\]
which contradicts \eqref{Positive_epistasis}.

(iii): A triple mutant $\{i,j,k\}$ and a double mutant $\{k,l\}$ that overlap cannot both be peaks.

Suppose that $\{ i,j,k \}$  and $\{k, l \}$ are both peaks. Then $w_{ijkl}<w_{ijk}$
and $w_{k}<w_{kl}$.
It follows that
\[
w_{ijkl}+w_{k}  < w_{ijk} + w_{kl}
\]
which contradicts \eqref{Positive_epistasis}.
\end{proof}


\begin{theorem}\label{theoremsecond}
If a fitness landscape induces the staircase triangulation, then
there are at most four peaks if $L=4 \text{ or }5$.
\end{theorem}

\begin{proof}
The argument relies on the properties (i)-(iii) described in Lemma \ref{prop}.

{\emph{Case L=4:}}
If  a single mutant ${\{ i \} }$ is a peak, then 
other peaks would have to be of the form ${ \{ i, j, k \} }$ or ${ \{ i, j, k, l } \}$
by (i). Consequently, there are at most two peaks.
If a triple mutant is a peak, then there are at most two peaks by symmetry.

It only remains to consider the case when no single or 
triple mutants are peaks. If a double mutant ${ \{ i, j \}}$ is peak, then
other peaks would have to be of the types
$\emptyset$, $ \{ k, l \} $ or $ \{ i, j, k, l \}$ by (ii).
Consequently, there are at most four peaks, which completes the
argument for $L=4$.

{\emph{Case L=5:}}
The argument is similar to the case $L=4$. We subdivide the hypercube into layers of genotypes of the same cardinality $k$,
where $0 \leq k \leq L$. Layers $k$ and $L-k$ are equivalent by symmetry, and layer $k$ contains $L \choose k$ genotypes.
For $L=5$, statements (ii) and (iii) imply that there
can be at most two double or triple mutants that are peaks, or one double and one triple mutant. 
By symmetry, the case with two triple mutants need not be considered.
Additionally one peak can
be placed in layer 0 or 1 and another one in layer 4 or 5. Table \ref{maximal} lists all peak
patterns with the maximal number of peaks that are allowed based on Lemma \ref{prop} up to symmetries.
They all have 4 peaks.

\begin{table}
\centerline{
  \begin{tabular}{c|ccccc}
      $k$ & I & II & III \\
      \hline
      0 & 1 & 1 & 1 \\
      1 & 0 & 0 & 0 \\
      2 & 2 & 2 & 1 \\
      3 & 0 & 0 & 1 \\
      4 & 1 & 0 & 0 \\
      5 & 0 & 1 & 1 \\
      \hline
      total & 4 & 4 & 4 \\
      \hline
  \end{tabular}}
  \caption{Distinct peak patterns with the maximal number of peaks for fitness
landscapes with L = 5 that induce the staircase triangulation.}
\label{maximal}
\end{table}
\end{proof}

The following example shows that the bound in Theorem \ref{theoremsecond} is sharp for $L=4$.
\begin{example}
Consider a generic fitness landscape where
\[
w_{0000} \approx 20, w_{1100} \approx 4, w_{0011} \approx 4, w_{1111} \approx 20,
\]
\
\[
w_{1010} \approx 1, w_{1001} \approx 1, w_{0110} \approx 1, w_{0101} \approx 1,
\]
and  all other genotypes have fitness approximately 3.
The four genotypes $w_{0000}, w_{1100}, w_{0011}$ and $w_{1111} $ are peaks
and the landscape induces the standard staircase triangulations. (See Section \ref{circuits}
for computational details).
\end{example}
A similar argument shows that the bound for $L=5$ is sharp as well.

\begin{remark}
  Recent work by Daniel Oros has shown that the number of peaks of a fitness landscape that induces the staircase triangulation
  is at most 8, 9, and 16 for $L=6, 7$ and $8$, respectively \cite{Oros2022}. 
  \end{remark}

\smallskip
\noindent
    {\bf{Open question:}} What is the maximal number of peaks of a fitness landscape that induces the staircase triangulation
    for arbitrary $L$?

\bigskip

\subsection{Staircase triangulations in random fitness landscapes and empirical data}

The occurrence of staircase triangulations in random fitness landscapes with uniformly distributed fitness values was investigated by simulations as described in Section  \ref{Sec:Statistics}. For $L = 3$, roughly 1/8 of all landscapes induce a staircase triangulation (type 6 triangulation), see Figure~\ref{Fig:sublandscapes} and \cite{Srivastava}, but for $L = 4$ this triangulation appears to be exceedingly rare. Among 25,000 random samples we could not find a single instance of the standard staircase triangulation, which is not surprising since there are 87,959,448 distinct regular triangulations for $L = 4$ (Table \ref{table:peakpatterns}). 

To assess how common staircase triangulations are in biological fitness landscapes, we analysed the protein fitness landscape presented in \cite{Wu2016}. This landscape comprises all 160,000 sequences that can be generated by varying four loci of protein G domain B1 (GB1) using all 20 amino acid alleles at each site. GB1 is an immunoglobulin-binding protein expressed in Streptococcal bacteria and has a total of 56 amino acids. The four chosen sites contain 12 of the top 20 positive epistatic interactions among all pairwise interactions. The fitness of each sequence was determined by both its stability (i.e., the fraction of folded proteins) and its function (i.e., binding affinity to immunoglobulin) and was measured by coupling mRNA display with Illumina sequencing. The distribution of the fitness is extremely skewed, with the majority of sequences having a very small fitness, while very few sequences having a very high fitness.

\begin{figure}
\begin{center}
\includegraphics[width=7.cm]{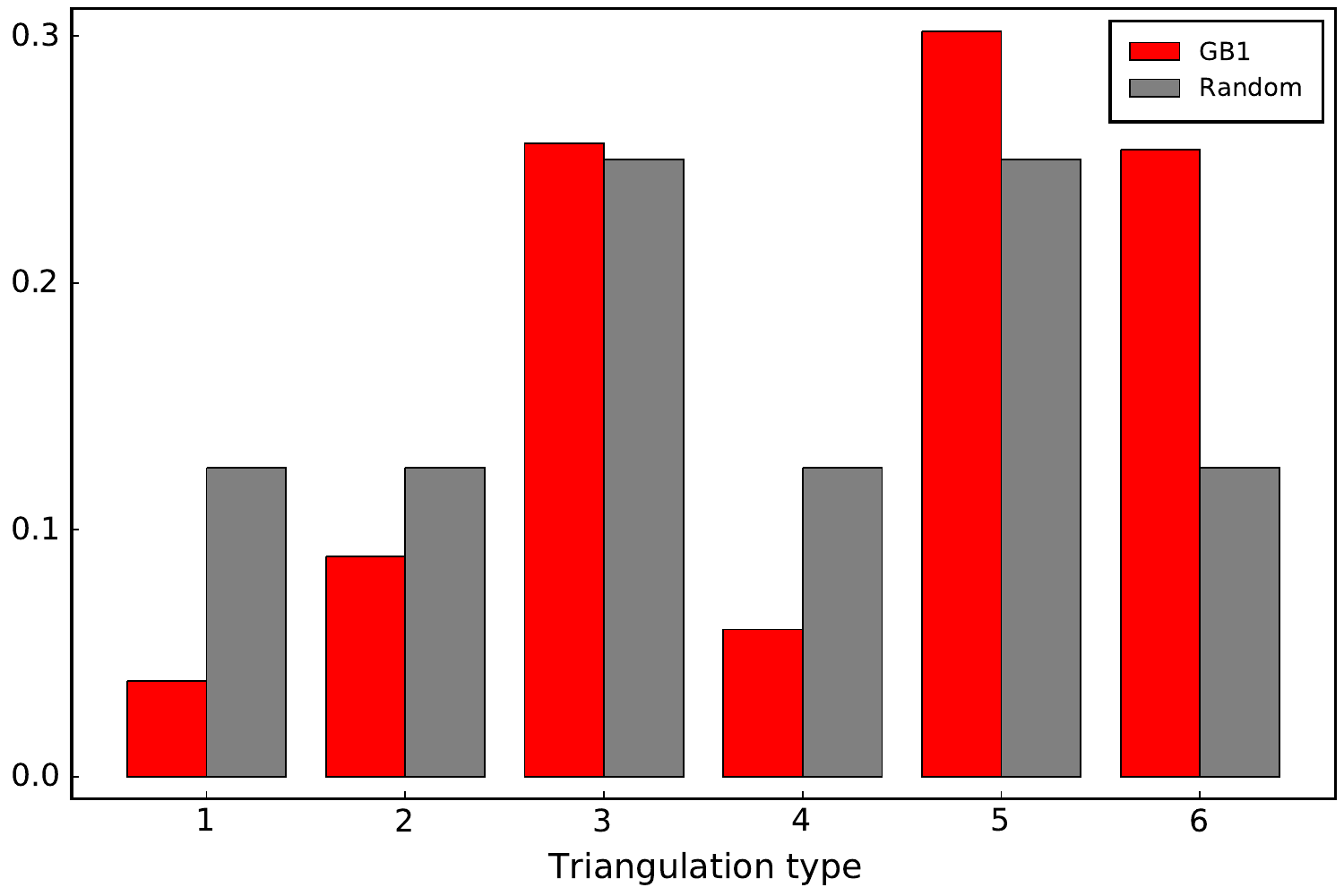}
\end{center}
\caption{Distribution of triangulation types for 3-locus random fitness landscapes with uniformly distributed fitness (grey bars)
  and biallelic 3-locus sublandscapes of the GB1 protein fitness landscape of \cite{Wu2016}. The probabilities for the random landscapes are close to the simple rationals $\frac{1}{8}, \frac{1}{8}, \frac{1}{4}, \frac{1}{8},
\frac{1}{4}, \frac{1}{8}$ \cite{Srivastava}.}\label{Fig:sublandscapes}
\end{figure}

Since the present analysis is focused on bi-allelic fitness landscapes, we examined bi-allelic sub-landscapes of the protein landscape. We started from any arbitrary sequence in the genotype space and then considered all the $2^4 = 16$ sequences that can be generated by mutating each locus to a different amino acid. For each locus, there are 19 other possibilities to mutate to, so the number of such sub-landscapes is enormously large. Therefore, we limited our analysis to 10,000 such sub-landscapes. The triangulation imposed by each sub-landscape was computed using the software Macaulay 2. For triangulations that displayed the maximal number of 24 simplices, we checked if the vertices of each of the simplices represented paths from the zero-string 0000 to its antipodal sequence 1111. Upon doing so, we found 10 landscapes that showed the standard staircase triangulation and all of them had either one or two peaks. From this, we can estimate the frequency of any staircase triangulation by multiplying by 8. Thus, we would expect to find around 80 landscapes that show staircase triangulation. Therefore, staircase triangulations are more common in the biological fitness landscape than in random fitness landscapes with completely uncorrelated fitness values. 

We additionally determined the distribution over triangulation types of 10,000 3-locus sub-landscapes in the GB1 data and compared this with the corresponding distribution for 10,000 random landscapes (Figure \ref{Fig:sublandscapes}). The fraction of GB1 landscapes of types 1, 2 and 4 are reduced in comparison to random landscapes, whereas those of types 3, 4 and 6 are enriched. Out of these, the enrichment of landscapes of type 6 is most notable. The strong over-representation of staircase triangulations in the biological sub-landscapes with 3 and 4 loci confirms the link of positive epistatic interactions to universal positive epistasis and demonstrates the biological relevance of this concept.

\section{Circuits,  Gr\"obner bases and the staircase triangulation}\label{circuits}
A triangulation can be described explicitly by a list of inequalities, or by  its minimal non-faces
(Chapter 9.4 \cite{triangulations}).
In brief, let  $J$ be a subset of the vertices of the $L$-cube. Then $J$ is a 
minimally dependent point set if there is an affine dependence relation 
$\sum_{g \in J}\lambda_g g = 0$, $\lambda_v\in\mathbb R$, with $\sum_{g \in J}\lambda_g=0$ and 
$J$ is minimal with this property. 
The expression $\sum_{g \in J}\lambda_g w_g$ is called a circuit, or a circuit interaction.
A minimal non-face is a set of genotypes that does not belong to any simplex in the triangulation,
such that the set is minimal with this property.

One can compare the  descriptions for the standard staircase triangulation  (Definition \ref{staircase}).
The minimal non-faces for $L=3$ are

\smallskip
$\{ 100, 010 \}, \{100, 001 \},  \{100, 011 \}  $,

$\{ 010, 001 \},  \{ 010, 101 \},  \{ 001, 110 \} $,

$\{ 110, 011 \},  \{ 110, 101 \},    \{ 101, 011 \}  $.

\smallskip
The triangulation can also be defined by the circuit signs:

$w_{000}+w_{110}-w_{100}-w_{010}>0$

$w_{000}+w_{101}-w_{100}-w_{001}>0$

$w_{000}+w_{011}-w_{010}-w_{001}>0$

$w_{100}+w_{111}-w_{110}-w_{101}>0$

$w_{010}+w_{111}-w_{110}-w_{011}>0$

$w_{001}+w_{111}-w_{101}-w_{011}>0$.

\smallskip
The six circuit signs imply the inequalties:

$w_{111}+w_{000}-w_{110}-w_{001}>0$

$w_{111}+w_{000}-w_{101}-w_{010}>0$

$w_{111}+w_{000}-w_{011}-w_{100}>0$.

The nine inequalities listed correspond exactly to all minimal non-faces (the two rightmost 
genotypes that appear on each line  are the genotypes in the minimal non-faces).

\begin{lemma}\label{non-faces}
The minimal non-faces for the staircase triangulation consist of the sets
\[
\{ g,  g' \}  \text{ such that } g' \not \subset g, g  \not \subset g' .
\]
\end{lemma}
\begin{proof}
Each pair $ \{ g,  g' \}$  as above is a minimal non-face.
Let $S$ be the set of such pairs.
 Let
$
\{ g_1, \dots, g_r \}
$
be a minimal non-face. 
The genotypes $g_1, \dots, g_r$ do not all belong to the same simplex, which excludes that 
the genotypes form a chain 
$
g_{1} \subseteq  \dots \subseteq g_{r}.
$
It follows that $g_i  \not \subset  g_j$ and  $g_j  \not \subset  g_i$ for some pair $i, j$,
and consequently that $\{g_i, g_j \}$ is a non-face. 
One concludes that
$
\{ g_1, \dots,  g_r \} =\{g_i, g_j \},
$
which implies that
$
\{ g_1, \dots, g_r \} \in S.
$
\end{proof}

\begin{observation}\label{positive}
If a fitness landscape has universal positive epistasis,
then the landscape induces the staircase triangulation.
\end{observation}

\begin{proof}
It is sufficient to show that the triangulation induced by the fitness landscape 
has the same minimal non-faces as the staircase triangulation.
Let $S$ be the set of pairs
\[
\{ g,  g' \}  \text{ such that } g' \not \subset g, g  \not \subset g'.
\]
By assumption, the landscape has  universal positive epistasis so that
\[
w_{g \cup g'} + w_{g \cap g'} \geq  w_{g} + w_{g'}
\]
 for any two genotypes $g,g'$.
If $ \{ g,  g' \} \in S$, the inequality is strict. It follows that every element in $S$
is a minimal non-face.  
Moreover, a non-face of the triangulation cannot consist of
genotypes that constitute a chain. 
By an argument similar to the proof of the 
previous result, it follows that $S$ is the set 
of minimal non-faces.
Since a landscape with universal positive epistasis
has the same minimal non-faces as the staircase triangulation,
the result follows.
\end{proof}

One can use Gr\"obner bases  \cite{froberg} for finding all minimal
non-faces of a triangulation (Chapter 9.4 \cite{triangulations}).

\begin{example}
Let $k[x_{00}, x_{10}, x_{01}, x_{11}]$ and $k[x_0, x_1, y_0, y_1]$ be polynomial rings, and let  $\varphi$ be the map defined by 
\[
\varphi:  \quad \quad x_{00} \mapsto x_0 y_0,  \quad x_{10} \mapsto x_1 y_0  \quad x_{01} \mapsto  x_0 y_1,  \quad x_{11}  \mapsto x_1 y_1 
\]
Let $I={\text{ ker }} \varphi$. Then $I$ is the ideal generated by $x_{11} x_{00}-x_{10} x_{01}$.

Let 
\[
w=(w_{00}, w_{10}, w_{01}, w_{00})
\] 
be a weight vector and let $J= in_{-w} (I)$ be the initial ideal with respect to $-w$ (notice the sign here, $-w$ defines the monomial order).
If 
\[ w_{11}+w_{00} - w_{10}-w_{01}>0,
\]
then $J=(x_{10} x_{01})$.

According to theory on Gr\"obner bases and triangulations,
 the generators of $J$ correspond exactly to the  minimal non-faces of the triangulation.
One concludes that  $\{ 10, 01 \}$ is the only minimal non-face and therefore that the 
triangulation induced by $w$ consists of the simplices
 $\{00, 10, 11\}$ and $\{00, 01, 11 \}$.
\end{example}

The following result follows immediately from Sturmfels' correspondence, see Theorem 9.4.5 in 
\cite{triangulations}.

\begin{lemma}
Let $k[x_g]$ be the polynomial ring in variables labeled by the genotypes $x_g$ and let $w$ be the weight vector that defines  the fitness landscape.
Let $I \subset k[x_g]$ be the defining ideal (as in the previous example).
If $ J=in_{-w} (I)$ is square free, then the generators of $J$ correspond exactly to the minimal non-faces of the triangulation.
\end{lemma}

\smallskip
The following observations are useful for computational purposes.
The results can be deduced from general theory on binomial ideals \cite{sturmfelsb}.
We provide an explicit argument for the readers' convenience.
In addition to the application here (see below) the result should be useful
for research on triangulations and peak patterns for  higher $L$.

\begin{observation}
Let $S$ be the set consisting of all elements of the form
\[
x_g \cdot x_{g'} - x_{g \lor g'}  \cdot   x_{g \wedge g'},
\]
where  $g$ and  $g'$  are genotypes such that $ g' \not \subset g$ and $ g  \not \subset  g'$.
Consider the monomial order defined by a fitness landscape that induces
the standard staircase triangulation. With notation as in Example 2,
the set $S$ is a  Gr\"obner basis for the defining ideal $I \subset k[x_g]$ .
\end{observation}

\begin{proof}
Let $I$ be the defining ideal and let $J=in_{-w} (I)$.
It is easy to verify that $S$ generates $I$. 
Let $x_{g_1} \cdots x_{g_k}$ be a generator for $J$.
It is sufficient to show that the generator is divisible by a leading term of $S$.

By the previous lemma, the set $\{ g_1, \dots, g_k \}$
is a minimal non-face of the triangulation. By Lemma \ref{non-faces},
 \[
 \{ g_1, \dots, g_k \}=\{g_i, g_j \} 
 \]
where $g_i  \not \subset  g_j$ and  $g_j  \not \subset  g_i$.
By construction, the element
\[
x_{g_i} x_{g_j}  - x_{g_i \lor g_j}  \cdot x_{g_i \wedge g_j}   \in S.
\]
It follows that $x_{g_1} \cdots x_{g_k}$ is divisible by $x_{g_i} x_ {g_j}$
and consequently  that   
\[
x_{g_1} \cdots x_{g_k}=x_{g_i} x_ {g_j}.
\]
\end{proof}

\begin{observation}
For $S$ as defined above,
$ \vert S \vert =2^{L-1}(2^L+1)-3^L$.
By construction,  $\vert S \vert$ equals the number of minimal
non-faces.  
\end{observation}

\begin{proof} According to Lemma \ref{non-faces}, $\vert S \vert$ is the number of pairs of genotypes $g,g'$ that are neither subsets nor
  supersets of each other. A genotype $g$ with $k$ 1's has $2^k-1$ subsets and $2^{L-k}-1$ supersets. The number of
  other genotypes that are neither subsets nor supersets of $g$ is therefore
  \[
  n_k = 2^L-1-(2^k-1+2^{L-k}-1) = 2^L - 2^k-2^{L-k}+1.
  \]
  Multiplying this with the number of genotypes of size $k$, summing over $k$ and dividing by 2 to avoid overcounting of pairs one arrives at
  \[
  \vert S \vert = \frac{1}{2} \sum_{k=1}^{L-1} {L \choose k}n_k = 2^{L-1}(2^L+1)-3^L.
  \]
  \end{proof}

For $L=4$ the Gr\"obner basis described above can be given explicitly (see below).
The result  was applied for verifying that the fitness landscape in Example 1 induces the staircase triangulation.

\newpage

{\emph{Gr\"obner basis for $L=4$:}}

\smallskip

$x_{1011} \cdot x_{0111}-x_{0011} \cdot  x_{1111}, \quad  x_{1101}  \, x_{0111}-x_{0101}  \cdot x_{1111}$,

$x_{1110} \cdot x_{0111}-x_{0110}  \cdot x_{1111}, \quad  x_{1101}  \cdot x_{1011}-x_{1001} \cdot  x_{1111}$,

$x_{1110} \cdot x_{1011}-x_{1010} \cdot x_{1111}, \quad x_{1110} \cdot x_{1101}-x_{1100} \cdot x_{1111}$,

$x_{1001} \cdot x_{0111}- x_{0001} \cdot x_{1111},  \quad  x_{1010} \cdot x_{0111}-x_{0010} \cdot  x_{1111}$,

$x_{1100} \cdot x_{0111}-x_{0100}  \cdot x_{1111}, \quad x_{0101}  \cdot x_{1011}-x_{0001} \cdot  x_{1111}$,

$x_{0110} \cdot x_{1011}-x_{0010} \cdot  x_{1111},   \quad  x_{1100} \cdot x_{1011}-x_{1000}  \cdot x_{1111}$,

$x_{0011} \cdot  x_{1101}-x_{0001} \cdot  x_{1111}, \quad  x_{0110} \cdot x_{1101}-x_{0100}  \cdot x_{1111} $

$x_{1010} \cdot  x_{1101}-x_{1000} \cdot  x_{1111},  \quad x_{0011} \cdot  x_{1110}-x_{0010}  \cdot x_{1111} $,

$x_{0101} \cdot  x_{1110}-x_{0100} \cdot  x_{1111},  \quad x_{1001} \cdot x_{1110}-x_{1000}  \cdot x_{1111} $,

$x_{1000} \cdot  x_{0111}-x_{0000} \cdot x_{1111},  \quad  x_{0100} \cdot  x_{1011}-x_{0000} \cdot x_{1111}$,

\smallskip
$x_{0010}  \cdot x_{1101} - x_{0000} \cdot x_{1111}, \quad  x_{0001} \cdot x_{1110}-x_{0000}  \cdot x_{1111} $, 

$x_{0101}  \cdot x_{0011} - x_{0001} \cdot x_{0111},  \quad x_{0110} \cdot x_{0011}-x_{0010}  \cdot x_{0111} $, 

$ x_{1001}  \cdot x_{0011} -x_{0001} \cdot x_{1011},   \quad x_{1010}  \cdot x_{0011}-x_{0010} \cdot x_{1011} $,
 
$ x_{1100}  \cdot x_{0011}-x_{0000} \cdot x_{1111},  \quad x_{0110}   \cdot x_{0101}-x_{0100} \cdot x_{0111} $,
 
$ x_{1001} \cdot x_{0101}-x_{0001} \cdot  x_{1101},  \quad x_{1010} \cdot x_{0101}-x_{0000}  \cdot x_{1111} $,

$x_{1100} \cdot x_{0101}-x_{0100} \cdot x_{1101},  \quad x_{1001}  \cdot x_{0110}-x_{0000} \cdot x_{1111}$,
  
$x_{1010} \cdot x_{0110}-x_{0010}  \cdot x_{1110}, \quad x_{1100} \cdot  x_{0110}-x_{0100}  \cdot x_{1110} $,
 
$x_{1010} \cdot x_{1001}-x_{1000} \cdot x_{1011}, \quad x_{1100}  \cdot x_{1001}-x_{1000}  \cdot  x_{1101}$,
 
$x_{1100} \cdot x_{1010}-x_{1000}  \cdot x_{1110}, \quad x_{0100}  \cdot x_{0011}-x_{0000}  \cdot x_{0111}$, 

$x_{1000} \cdot x_{0011}-x_{0000}  \cdot x_{1011}, \quad  x_{0010}  \cdot  x_{0101}-x_{0000} \cdot x_{0111}$,

\smallskip

$x_{1000} \cdot  x_{0101}-x_{0000} \cdot x_{1101}, \quad  x_{0001} \cdot x_{0110}-x_{0000} \cdot x_{0111} $

$x_{1000} \cdot x_{0110}-x_{0000} \cdot x_{1110},  \quad x_{0010} \cdot x_{1001}-x_{0000} \cdot  x_{1011}$, 

$x_{0100} \cdot x_{1001}-x_{0000}  \cdot x_{1101},  \quad x_{0001} \cdot x_{1010}-x_{0000} \cdot x_{1011}$, 

$x_{0100} \cdot x_{1010}-x_{0000}  \cdot x_{1110}, \quad  x_{0001} \cdot  x_{1100}-x_{0000}  \cdot x_{1101}$, 

$x_{0010}  \cdot x_{1100}-x_{0000} \cdot  x_{1110},  \quad x_{0010} \cdot x_{0001}-x_{0000} \cdot x_{0011}$

\smallskip

$x_{0100} \cdot x_{0001} - x_{0000} \cdot x_{0101},  \quad   x_{1000} \cdot x_{0001}-x_{0000} \cdot x_{1001}$,

$x_{0100} \cdot x_{0010}-x_{0000} \cdot x_{0110},   \quad x_{1000} \cdot x_{0010}-x_{0000} \cdot x_{1010}$,
 
$x_{1000} \cdot x_{0100}-x_{0000} \cdot x_{1100} $

\bmhead{Acknowledgments}
This work was initiated by a mini-symposium at the 2019 SIAM Conference on Applied Algebraic Geometry. We thank Lisa Lamberti for organizing
this symposium, for many useful discussions, and for a careful reading of our
manuscript. We also thank Muhittin Mungan, Daniel Oros and Anna Seigal for helpful
remarks, and two anonymous reviewers for their constructive suggestions.

\begin{center}
\begin{table}
\renewcommand{\arraystretch}{2}
\begin{tabular}{|p{3cm}|p{3.7cm}|p{5cm}|}\hline
{\bf Notation}  & {\bf Mathematical term} &  {\bf Biological term} \\ \hline
$\{0,1\}^L$ & Bit strings of length $L$ & The set of $2^L$ genotypes in a {\bf{biallelic $L$-locus}} system  \\ \hline
$g_i$ & $i$:th position of the bit string $g$ &  {\bf{Locus $i$}} of genotype $g$ \\ \hline
$g_i=1$ &  &  $g$ has a 1-{\bf{allele}} at locus $i$  \\ \hline
$[0,1]^L$ & Unital $L$-dimensional hypercube & Genotope for a biallelic $L$-locus system (terminology used in the shape theory)\\ \hline
$g, g' \in \{0,1\}^L : |g-g|=d$  & Vertices with distance $d$ in the undirected cube graph  &  Genotypes with  {\bf{Hamming distance $d$}}  \\ \hline
                                                    &  If $d$ (as above) is 1   & The genotypes $g$ and $g'$ are {\bf{mutational neighbors}}  \\ \hline
$w: \{0,1\}^L \to \mathbb{R}_{\geq 0} ; g \mapsto w(g)=w_g$
&      & {\bf{Fitness landscape}} \\ \hline
&  Acyclic directed hypercube graphs induced by  fitness landscapes & {\bf{Fitness graphs:}} the vertices represent genotypes and edges between mutational neighbors are directed towards genotypes with higher fitness.\\ \hline
& A vertex $g$ is a {\bf{sink}}  in the fitness graph if all edges  incident to $g$ are directed towards  $g$.  
&  A genotype $g$ is a {\bf{peak}} in the fitness landscape (and fitness graph) if all mutational neighbors of $g$ have strictly lower  fitness than  $g$.  \\ \hline
 & The triangulation of $[0,1]^L$ induced by the fitness landscape & The {\bf{shape}} of the fitness landscape \\ \hline
\end{tabular}
\caption{Dictionary of commonly used terms.}\label{notation}
\end{table}
\end{center}

\clearpage


\bibliography{shapesfreferences}


\end{document}